\documentclass[12pt]{iopart}
\usepackage{iopams,amsthm}

\newcommand{\dv}{\,\mathrm{div}\,}
\newcommand{\rot}{\mathrm{rot}\,}
\newcommand{\const}{\mathrm{const}}
\newcommand{\bu}{\mathbf{u}}
\newcommand{\bx}{\mathbf{x}}

\newcommand{\bb}{\mathbf{b}}
\newcommand{\bbmu}{\bar{\bmu}}

\newcommand{\tbepsilon}{\widetilde{\bepsilon}}
\newcommand{\tbkappa}{\widetilde{\bkappa}}

\newcommand{\bbnu}{\bar{\bnu}}
\newcommand{\bA}{\mathbf{A}}

\newcommand{\ba}{\mathbf{a}}
\newcommand{\bk}{\mathbf{k}}

\newcommand{\bd}{\mathbf{d}}

\newcommand{\bB}{\mathbf{B}}

\newcommand{\bv}{\mathbf{v}}
\newcommand{\be}{\mbox{\boldmath $e$}}
\newcommand{\pd}[2]{\frac{\partial #1}{\partial #2}}

\newcommand{\pdd}[3]{\frac{\partial^2 #1}{\partial #2 \partial #3}}

\newcommand{\tbv}{\widetilde{\bf v}}
\newcommand{\tbB}{\widetilde{\bf B}}

\theoremstyle{plain}
\newtheorem{theorem}{Theorem}

\newtheorem{lemma}{Lemma}

\theoremstyle{definition}

\theoremstyle{remark}

\newtheorem{remark}{Remark}

\eqnobysec

\begin{document}

\begin{center}
\bf
COMPLETE CLASSIFICATION OF STATIONARY FLOWS WITH CONSTANT TOTAL PRESSURE OF IDEAL INCOMPRESSIBLE INFINITELY CONDUCTING FLUID

\bigskip

S.V. Golovin$^{1,2}$, M.K. Krutikov$^2$

\bigskip

$^1$ Lavrentyev Institute of Hydrodynamics SB RAS\\
$^2$ Novosibirsk State University

\rm
\end{center}

\bigskip
\small
\parbox[b]{14.5cm}{The exhaustive classification of stationary incompressible flows with constant total pressure of ideal infinitely electrically conducting fluid is given. By introduction of curvilinear coordinates based on streamlines and magnetic lines of the flow the system of magnetohydrodynamics (MHD) equations is reduced to a nonlinear vector wave equation extended by the incompressibility condition in a form of a generalized Cauchy integral. For flows with constant total pressure the wave equation is explicitly integrated, whereas the incompressibility condition is reduced to a scalar equation for functions, depending on different sets of variables. The central difficulty of the investigation is the separation of variables in the scalar equation, and integration of the resulting overdetermined systems of nonlinear partially differential equations. The canonical representatives of all possible types of solutions together with equivalence transformations, that extend the canonical set to the whole amount of solutions are found.}

\normalsize

%\tableofcontents

\bigskip
\section*{Introduction}
The mathematical model of ideal magnetohydrodynamics (MHD) describes macroscopic motions of infinitely electrically conducting fluid under the action of internal pressure, magnetic and inertia forces. The area of applicability of the model is extremely large: from problems of magnetic confinement to astrophysics. The popularity of the model is explained by its relative simplicity, by the wealth of its mathematical content, and by the variety of described physical effects. In most applications it is required to have solutions that model essentially three-dimensional motions of fluid. At that, the numerical or analytical analysis is complicated by the nonlinearity of equations. In this respect exact solutions that describe three-dimensional flows of conducting fluid are of interest.

Advancement in the construction of exact solutions can be reached by introduction of special curvilinear system of coordinates  \cite{Golovin2010}, determined by the geometry of streamlines and magnetic lines. In this form MHD equations are partially integrated and reduced to a vector wave equation extended by the incompressibility condition in the from of the generalized Cauchy integral. Any solution to the equations immediately gives explicit description of magnetic lines and streamlines of the flow.

This work is devoted to complete description of ideal electrically conducting fluid flows with constant total pressure $p+\frac{1}{2}|\bB|^2=\const$. For this class of solutions the vector wave equation is integrated explicitly. Substitution of the obtained representation of solution into the Cauchy integral yields a nonlinear scalar equation, that contains functions, depending on different sets of arguments. Separation of variables in the scalar equation, and subsequent integration of obtained overdetermined systems of partially differential equation is the main difficulty of the research. The significant simplification of calculations is reached by the systematic use of equivalence transformations, which allow one to discard insufficient arbitrary functions and constants. The separation of variables performed in Theorem \ref{T1} produces a set of 9 overdetermined systems of equations. Integration of each system with the use of equivalence transformations leads to the formulated in Theorem \ref{T2} main result of the paper: the complete list of canonical forms of solutions, describing MHD flows with constant total pressure. It is shown that in all such flows Maxwellian surfaces bearing streamlines and magnetic lines of the flow belong to a class of translational surfaces, i.e. are swept by the parallel translation of some 3D curve, called generatrix, along another 3D curve called directrix. In Theorem \ref{T3} it is proved that all solutions belong to either of two main classes: flows with cylindrical Maxwellian surfaces, or flows with Maxwellian surfaces having a common directrix.

Class of stationary ideal MHD flows with constant total pressure was earlier investigated in \cite{Golovin2010}. Here the curvilinear system of coordinates was introduced, general properties of solutions were investigated, some exact solutions describing motions with constant total pressure were obtained. In \cite{Schief2003} with the use of a similar curvilinear coordinates there were obtained exact solutions, where pressure is constant not in the whole area occupied by the flow, but only on toroidal Maxwellian surfaces. Classes of solutions with constant total pressure in case of non-stationary flows were obtained in \cite{Golovin2011}. In particular, solutions describing flows with knotted or toroidal Maxwellian surfaces were constructed. Note, that problem of complete description of such solutions in non-stationary case is technically more complicated and is not solved yet.

\section{Preliminary information} Equations of ideal magnetohydrodynamics (MHD), describing stationary flows of ideal infinitely conducting fluid with frozen-in magnetic field have the following dimensionless form:
\begin{equation}\label{SBMHD:MHD}
\begin{array}{l}
\rho(\bu\cdot\nabla)\bu+\bB\times\rot\bB+\nabla p=0,\;\;\;\rot(\bu\times\bB)=0,\\[2mm]
\dv\bu=0,\;\;\;\dv\bB=0.
\end{array}
\end{equation}
Here $\bu$ is the velocity, $\bB$ is the magnetic field, $\rho$ is the density, $p$ is the pressure. Density $\rho$ is assumed to conserve along streamlines: $\bu\cdot\nabla\rho=0$. The following functions are introduced:
\begin{equation}\label{NewVars}
\bv=\bu\,\sqrt{\rho(\bx)},\;\;\;P=p+\frac{1}{2}\,\bB^2
\end{equation}
Function $P$ will be referred to as the total pressure. With these functions system (\ref{SBMHD:MHD}) can be rewritten in the following form:
\begin{equation}\label{SBMHD:MHD1}
\begin{array}{l}
(\bv\cdot\nabla)\,\bv-(\bB\cdot\nabla)\,\bB+\nabla P=0,\\[2mm]
(\bB\cdot\nabla)\bv=(\bv\cdot\nabla)\bB,\\[2mm]
\dv\bv=0,\;\;\;\dv\bB=0.
\end{array}
\end{equation}
Equation for density $\rho$ splits off from the system (\ref{SBMHD:MHD1}) and can be solved separately. Equations
(\ref{SBMHD:MHD1}) form an overdetermined compatible system of 8 equations for 7 sought functions.

Except for the obvious symmetry group of transformations consisting of shifts along the coordinate axes, rotations, dilatations and shift of the complete pressure, system of equations (\ref{SBMHD:MHD1}) admits the following transformation of velocity and magnetic vector fields \cite{Bogoyavl2002, Schief2003}:
\begin{equation}\label{SBMHD:Transform}
\tbv=\bB\sinh f(\varphi)+\bv\cosh f(\varphi),\quad\tbB=\bB\cosh f(\varphi)+\bv\sinh f(\varphi).
\end{equation}
Here $\varphi$ is an arbitrary function that conserves along streamlines and magnetic lines of the flow:
\begin{equation}\label{SBMHD:Phi}
(\bv\cdot\nabla)\varphi=0,\;\;\;(\bB\cdot\nabla)\varphi=0.
\end{equation}
The symmetry group of system (\ref{SBMHD:MHD1}) will be used for simplification of the form of obtained solutions.

\section{Curvilinear system of coordinates} System of equations (\ref{SBMHD:MHD1}) has more symmetric form in terms of the following vector fields, also known as Elsa\"{a}sser variables \cite{Elsaesser1950}:
\begin{equation}\label{SBMHD:1}
\ba=\bv-\bB,\quad\bb=\bv+\bB
\end{equation}
or
\begin{equation}\label{SBMHD:2}
\bv=\frac{1}{2}(\bb+\ba),\quad\bB=\frac{1}{2}(\bb-\ba).
\end{equation}
In terms of these vector fields system (\ref{SBMHD:MHD1}) is transformed to
%\begin{subequations}\label{SBMHD:MHD2}
\numparts
\begin{eqnarray}\label{SBMHD:3}
&&(\ba\cdot\nabla)\bb+\nabla P=0,\\[2mm]\label{SBMHD:4}
&&(\ba\cdot\nabla)\bb=(\bb\cdot\nabla)\ba,\\[2mm]\label{SBMHD:5}
&&\dv\ba=0,\;\;\;\dv\bb=0.
\end{eqnarray}
\endnumparts
%\end{subequations}

Note that equation (\ref{SBMHD:4}) is equivalent to the equality to zero of the commutator of vector fields $\ba$ and $\bb$:
\begin{equation}\label{SBMHD:7}
[\ba,\bb]:=(\ba\cdot\nabla)\bb-(\bb\cdot\nabla)\ba=0.
\end{equation}
Under the assumption of linear independence of these vector fields, the commutativity condition (\ref{SBMHD:7}) implies that the fields can be chosen as basic ones for a certain curvilinear system of coordinates \cite{Schutz1980en}. Indeed, let us observe a curvilinear system of coordinates $(k^1,k^2,k^3)$ such that $k^1$- and $k^2$-coordinate lines coincide with integral curves of vector fields $\ba$ and $\bb$ accordingly. In other words, in the space $\mathbb{R}^3$ the following non-degenerate mapping is introduced
\begin{equation}\label{SBMHD:8}
\bx=\bx(\bk), \quad\det\left(\pd{\bx}{\bk}\right)\ne0
\end{equation}
such that
\begin{equation}\label{SBMHD:9}
\ba=\frac{\partial\bx}{\partial k^1},\quad\bb=\frac{\partial\bx}{\partial k^2}.
\end{equation}
In this representation equation (\ref{SBMHD:7}) is satisfied automatically as a compatibility condition of equations (\ref{SBMHD:9}) for unknown vector $\bx(\bk)$. The remaining equations (\ref{SBMHD:3}), (\ref{SBMHD:5}) lead to the following system of equations for unknown functions $\bx=\bx(\bk)$, $P=P(\bk)$ (details can be found in \cite{Golovin2010}):
\begin{equation}\label{SBMHD:Main}
\fl
\begin{array}{l}
\displaystyle \pdd{\bx}{k^1}{k^2}+
\pd{P}{k^1}\left(\pd{\bx}{k^2}\times\pd{\bx}{k^3}\right) +
\pd{P}{k^2}\left(\pd{\bx}{k^3}\times\pd{\bx}{k^1}\right) +
\pd{P}{k^3}\left(\pd{\bx}{k^1}\times\pd{\bx}{k^2}\right)
=0,\\[4mm]
\displaystyle \pd{\bx}{k^1}\cdot\left(\pd{\bx}{k^2}\times\pd{\bx}{k^3}\right)=1.
\end{array}
\end{equation}
Henceforth `` $\cdot$ '', and `` $\times$ '' denote the Euclidean dot and vector products in $\mathbb{R}^3$ correspondingly. The first equation in (\ref{SBMHD:Main}) can be treated as a nonlinear vector wave equation for components of vector $\bx$. The second equation in (\ref{SBMHD:Main}) represents an incompressibility condition of both fluid an magnetic field. System of equations (\ref{SBMHD:Main}) is equivalent to the system (\ref{SBMHD:3})--(\ref{SBMHD:5}) on solutions with non-collinear vector fields $\ba$ and $\bb$.

Parametrization of the solution in terms of the dependence $\bx(\bk)$ gives the explicit description of the geometry of the fluid motion. For example, important in applications Maxwellian surfaces bearing magnetic field lines and streamlines of the flow, in this parametrization are given by $\bx=\bx(k^1,k^2,k^3_0)$ with arbitrary parameters $k^1$, $k^2$, and fixed parameter $k^3=k^3_0$. Each Maxwellian surface can be taken as an infinitely conducting wall bounding the area of the flow.

Note that system of equations (\ref{SBMHD:Main}) is invariant under the following change of parameters
\begin{equation}\label{SBMHD:BT}
\tilde{k}^1=\varphi(k^3) k^1,\quad\tilde{k}^2=\frac{k^2}{\varphi(k^3)},\quad \tilde{k}^3=k^3
\end{equation}
with an arbitrary function $\varphi$. This symmetry follows from the transformation (\ref{SBMHD:Transform}), admitted by the initial system (\ref{SBMHD:MHD1}). Besides, system of equations (\ref{SBMHD:Main}) is invariant under transformations
\begin{equation}\label{SBMHD:IT}
\tilde{k}^1=k^1+\psi(k^3), \quad\tilde{k}^2=k^2+\chi(k^3)
\end{equation}
with arbitrary functions $\psi$ and $\chi$. In what follows formulae (\ref{SBMHD:BT}), (\ref{SBMHD:IT}) will be used as equivalence transformations on the set of solutions of equations (\ref{SBMHD:Main}).

\section{Flows with constant total pressure} System (\ref{SBMHD:Main}) possesses a class of solutions where the total pressure $P$ is constant over the whole area occupied by the flow:
\[P=\const.\]
In this case the first equation of (\ref{SBMHD:Main}) is reduced to a linear one, and can be integrated explicitly as
\begin{equation}\label{SBMHD:14}
\pdd{\bx}{k^1}{k^2}=0\quad\Leftrightarrow\quad \bx=\bsigma(k^1,k^3)+\btau(k^2,k^3).
\end{equation}
Here $\bsigma$ and $\btau$ are vector functions depending on different independent variables. From (\ref{SBMHD:14}) it immediately follows that Maxwellian surfaces on this class of solutions are translational surfaces. They are spanned by parallel translation of the curve $\bx=\bsigma(k^1,k^3_0)$ along the curve $\bx=\btau(k^2,k^3_0)$ with a fixed parameter $k^3=k^3_0$.

Substitution of representations (\ref{SBMHD:14}) into the second of equations (\ref{SBMHD:Main}) gives the restriction for vectors $\bsigma$ and $\btau$:
\begin{equation}\label{SBMHD:15}
\pd{\bsigma}{k^1}\cdot\left[\pd{\btau}{k^2}\times\left(\pd{\bsigma}{k^3}+\pd{\btau}{k^3}\right)\right]=1.
\end{equation}
Now the description of flows with constant total pressure is reduced to a separation of variables $k^1$ and $k^2$ in equation (\ref{SBMHD:15}) and integration of obtained overdetermined systems of equations.

Henceforth the following notations are adopted. The standard Cartesian basis of orthonormal constant vectors in $\mathbb{R}^3$ is denoted by $\be^1$, $\be^2$, $\be^3$. The lower index $i$ denotes the derivative with respect to $k^i$ (e.g., $\bsigma_1\equiv \partial\bsigma/\partial k^1$). Bold Roman or Greek letters denote vectors.

\section{Separation of variables} Note that equations (\ref{SBMHD:15}) can be written as the equality to unity of the dot product of two vectors in $\mathbb{R}^6$, such that each vector depends on its own set of variables:
\begin{equation}\label{SBMHD:17}
\bA\cdot \bB=1,
\end{equation}
where
\begin{equation}\label{SBMHD:26}
\bA=(\bsigma_3\times\bsigma_1,\bsigma_1)^T,\quad\bB=(\btau_2,\btau_2\times\btau_3)^T
\end{equation}
Here $\bA=\bA(k^1,k^3)$, $\bB=\bB(k^2,k^3)$, symbol ``$T$'' denotes the transpose of a matrix or a vector. Separation of variables $k^1$ and $k^2$ in equation (\ref{SBMHD:17}) is managed by the following Lemma, valid for general $n$-dimensional vectors $\bA(x)$, $\bB(y)$ depending on different variables $x$ and $y$ \cite{Golovin2002en}:

\begin{lemma}\label{SBMHD:l2}
Equation (\ref{SBMHD:17}) is satisfied for vectors
\[\bA=\bigl(a^1(x),\ldots,a^n(x)\bigr)^T,\quad\bB=\bigl(b^1(y),\ldots,b^n(y)\bigr)^T\]
identically on variables $x$ and $y$ if and only if there exists a number $s$ $(0 \le s \le n-1)$, constant $n\times s$ matrix $C$ of rank $s$, constant vector $\bd\in\mathbb{R}^n$ and a set of $s$ non-constant linearly independent functions, forming a vector $\bu(x)$, such that, the following equations are satisfied:
\begin{equation}\label{SBMHD:25}
\bA=C\bu(x)+\bd;\quad \bB^TC=0,\quad \bB^T\bd=1.
\end{equation}
\end{lemma}
\begin{proof}
Let the linear span of vectors $\bA(x)$ for various admissible values of $x$ has a dimension $k$. This implies that there exist exactly $k$ various values $x^i$, $i=1,\ldots,k$, for which constant vectors $\bA^i=\bA(x^i)$ are linearly independent. An arbitrary vector $\bA(x)$ can be represented as a linear combination of the constant vectors as
\begin{equation}\label{Adecomp}
\bA(x)=\sum\limits_{i=1}^k u^i(x)\bA^i.
\end{equation}
By virtue of (\ref{Adecomp}) one has $u^i(x^j)=\delta^{ij}$ with the Kronecker delta in the right-hand side. This implies linear independence of functions $u^i(x)$. For the compatibility of equations (\ref{SBMHD:17}) one function among $u^i$ should be constant. Suppose, that this function has index $k$, i. e. $u^k\equiv1$. By arranging vectors $\bA^1,\ldots,\bA^{k-1}$ in a matrix $C$ with $n$ columns and $s=k-1$ rows and denoting $\bd=\bA^k$ one obtains a representation of vector $\bA(x)$, specified in the formulation of the Lemma. Equations for vector $\bB(y)$ follows from the substitution of the representation of vector $\bA$ into equation (\ref{SBMHD:17}) and splitting with respect to linearly independent functions $u^i(x)$.
\end{proof}

In what follows Lemma \ref{SBMHD:l2} is used as the main tool for separation of variables in equation (\ref{SBMHD:17}) with vectors $\bA$, and $\bB$ determined by (\ref{SBMHD:26}). At that, parameters $k^1$ and $k^2$ play the role of variables $x$, and $y$. Matrix $C$ and vector $\bd$ are supposed to depend arbitrarily on variable $k^3$. Only cases $s=0,1,2$ will be successively investigated. By virtue of the symmetry of equation (\ref{SBMHD:17}) with respect to the change of variables $k^1\leftrightarrow k^2$, $\bsigma\leftrightarrow\btau$, cases $s=3,4,5$ are equivalent to the observed ones.

In what follows, the following equivalence transformations of equation (\ref{SBMHD:17}) will be used for simplification of the resulting equations and solutions:
\begin{itemize}
\item $\widetilde{\bsigma}=\bsigma-\bepsilon(k^3)$ , $\widetilde{\btau}=\btau+\bepsilon(k^3)$ with arbitrary vector $\bepsilon$;
\item Transformations (\ref{SBMHD:BT}) and (\ref{SBMHD:IT}).
\end{itemize}

The analysis is preceded by the description of two characteristic cases of integration of equation (\ref{SBMHD:15}).

\section{Case of cylindrical Maxwellian surfaces}\label{CylMax} Suppose that vector $\bsigma$ depends linearly on variable $k^1$:
\[\bsigma=k^{1}\bfeta(k^{3}).\]
Substitution of this representation into equation (\ref{SBMHD:15}) gives
\[k^1\btau_2\cdot\bigl(\bfeta(k^3)\times\bfeta'(k^3)\bigr)+\bfeta(k^3)\cdot\bigl(\btau_2\times\btau_3\bigr)=1\]
Hereafter prime denotes the derivative with respect to the single variable $k^3$. Splitting on $k^1$ produces a system of equations for vector $\btau$:
\begin{equation}\label{Cyl1}
\btau_2\cdot\bigl(\bfeta'(k^3)\times\bfeta(k^3)\bigr)=0,\quad \bfeta(k^3)\cdot\bigl(\btau_2\times\btau_3\bigr)=1.
\end{equation}
The first of equations (\ref{Cyl1}) is satisfied identically if $\bfeta=\const$. In this case without loss of generality one can assume  $\bfeta=\be^1$. The exact solution can be written as
\begin{equation}\label{SBMHD:s0}
\bx=k^1\be^1+\btau(k^2,k^3),\quad
\left|\begin{array}{cc}
\tau^2_2&\tau^3_2\\[1mm]
\tau^2_3&\tau^3_3
\end{array}\right|=1.
\end{equation}

For a non-constant vector $\bfeta$ the first of equations (\ref{Cyl1}) gives
\[\btau=u(k^2,k^3)\bfeta'(k^3)+v(k^2,k^3)\bfeta(k^3)+\bdelta(k^3).\]
By substituting this representation into (\ref{Cyl1}) after obvious transformations one obtains an equation for function $u$ in the form
\[\bfeta\cdot\left[\pd{u}{k^2}\bfeta'(k^3)\times\bigl(u\bfeta''(k^3)+\delta(k^3)\bigr)\right]=1.\]
Integration with respect to $k^2$ leads to the following exact solution:
\begin{equation}\label{SBMHD:slin}
\bx=k^1\bfeta(k^3)+u(k^2,k^3)\bfeta'(k^3)+v(k^2,k^3)\bfeta(k^3)+\bdelta(k^3).
\end{equation}
Here $\bfeta(k^3)$, and $\bdelta(k^3)$ are arbitrary vectors, $v(k^2,k^3)$ is an arbitrary function. Function $u(k^2,k^3)$ is determined by the solution of the quadratic equation
\begin{equation}\label{SBMHD:slin1}
\frac{\bfeta\cdot(\bfeta'\times\bfeta'')}{2}\,u^2+\bigl(\bfeta\cdot(\bfeta'\times\bdelta')\bigr)u=k^2.
\end{equation}
Due to the linear dependence of vector $\bx$ on parameter $k^1$, Maxwellian surfaces in this solution are exhausted by cylinders with the generatrix directed along vector $\beta$. At that, the directrix of the cylinder is defined with a significant functional arbitrariness by formulae above.

\section{Maxwellian surfaces with a fixed directing curve} Equation (\ref{SBMHD:15}) also possesses the following class of solutions:
\numparts
\begin{eqnarray}\label{Sol21}
&&\bx=\bsigma\bigl(\lambda(k^1,k^3)\bigr)+\btau(k^2,k^3),\\[2mm]\label{Sol22}
&&\bsigma\bigl(\lambda(k^1,k^3)\bigr)\cdot\balpha(k^3)=k^1,\\[2mm]\label{Sol23}
&&\btau_2\times\btau_3=\balpha(k^3).
\end{eqnarray}
\endnumparts
Here $\bsigma\bigl(\lambda)$ and $\balpha(k^3)$ are arbitrary vectors. Scalar equation (\ref{Sol22}) implicitly defines function $\lambda(k^1,k^3)$. Equations (\ref{Sol23}) with fixed right-hand side form a closed system of differential equations for three components of vector $\btau$.

\begin{remark}\label{rem2}
Solution (\ref{Sol21})--(\ref{Sol23}) can be generalized in the following way: If one of components of vector $\bsigma$ is identically zero, then the corresponding component of vector $\balpha$ can be assumed to be a function of variables $k^2$, and $k^3$ defined by vector $\btau$ from the same component of equation (\ref{Sol23}). In this case the total number of equations (\ref{Sol23}) for vector $\btau$ reduces.
\end{remark}
In what follows we will refer to solution (\ref{Sol21})--(\ref{Sol23}) taking into account Remark \ref{rem2}. In this class of solutions Maxwellian surfaces have common generatrix given by equation $\bx=\bsigma(\lambda)$. Parametrization of this curve varies from one Maxwellian surface to another, however the generatrix itself remains unchanged.

\section{Separation of variables (\ref{SBMHD:15})} By using Lemma \ref{SBMHD:l2}, after partial integration and simplification of the equations, one obtains the following

\begin{theorem}\label{T1} Equation (\ref{SBMHD:15}) after the separation of variables can be brought by the equivalence transformation to one of the following non-equivalent forms:

$\underline{{\bf s=0}}$
\begin{equation}\label{Sep11}
\left\{\begin{array}{l}
\bsigma_{3}\times\bsigma_{1}=0,\quad\bsigma_{1}=\bnu;\\[2mm]
\bnu\cdot(\btau_{2}\times\btau_{3})=1.
\end{array}\right.
\end{equation}

$\underline{{\bf s=1}}$
\begin{equation}\label{Sep21}
\left\{\begin{array}{l}
\bsigma_{3}\times\bsigma_{1}=0,\quad\bsigma_{1}=U(k^{1},k^{3})\bepsilon+\bkappa;\\[2mm]
\bepsilon\cdot(\btau_{2}\times\btau_{3})=0,\quad\bkappa\cdot(\btau_{2}\times\btau_{3})=1.
\end{array}\right.
\end{equation}
\begin{equation}\label{Sep22}
\left\{\begin{array}{l}
\bsigma_3\times\bsigma_1=\bnu\times\bepsilon,\quad\bsigma_1=U(k^1,k^3)\bepsilon;\\[2mm]
\bnu\cdot(\bepsilon\times\btau_2)=1,\quad \btau_3\cdot(\bepsilon\times\btau_2)=0.
\end{array}\right.
\end{equation}
\begin{equation}\label{Sep23}
\left\{\begin{array}{l}
\bsigma_3\times\bsigma_1=U(k^1,k^3)\bmu\times\bkappa,\quad\bsigma_1=\bkappa;\\[1mm]
\bmu\cdot(\bkappa\times\btau_2)=0,\quad \btau_3\cdot(\bkappa\times\btau_2)=1.
\end{array}\right.
\end{equation}

$\underline{{\bf s=2}}$

\begin{equation}\label{Sep31}
\left\{\begin{array}{l}
\bsigma_3\times\bsigma_1=0,\quad\bsigma_1=U^1(k^1,k^3)\bepsilon^1 +U^2(k^1,k^3)\bepsilon^2 +\bkappa;\\[2mm]
      \bepsilon^1\cdot(\btau_2\times\btau_3)=0,\quad\bepsilon^2\cdot(\btau_2\times\btau_3)=0,
      \quad\bkappa\cdot(\btau_2\times\btau_3)=1.
\end{array}\right.
\end{equation}

\begin{equation}\label{Sep32}
\left\{\begin{array}{l}
\bsigma_3\times\bsigma_1=\bepsilon^1\times\bepsilon^2,\quad \bsigma_1=U^1(k^1,k^3)\bepsilon^1 +U^2(k^1,k^3)\bepsilon^2 ;\\[2mm]
      \btau_2\cdot(\bepsilon^1\times\bepsilon^2)=1,\quad\bepsilon^1\cdot(\btau_2\times\btau_3)=0,
      \quad\bepsilon^2\cdot(\btau_2\times\btau_3)=0.
\end{array}\right.
\end{equation}

\begin{equation}\label{Sep33}
\left\{\begin{array}{l}
\bsigma_3\times\bsigma_1=U^2(k^1,k^3)\bepsilon\times\bkappa,\quad \bsigma_1=U^1(k^1,k^3)\bepsilon +\bkappa ,\\[2mm]
      \btau_2\cdot(\bepsilon\times\bkappa)=0,\quad\bepsilon\cdot(\btau_2\times\btau_3)=0,\quad\bkappa\cdot(\btau_2\times\btau_3)=1.
\end{array}\right.
\end{equation}

\begin{equation}\label{Sep34}
\left\{\begin{array}{l}
\bsigma_3\times\bsigma_1=\bigl(U^2(k^1,k^3)\bmu+\bnu\bigr)\times\bepsilon,\quad\bsigma_1=U^1(k^1,k^3)\bepsilon;\\[2mm]
      \btau_2\cdot(\bmu\times\bepsilon)=0,\quad\btau_2\cdot(\bnu\times\bepsilon)=1,\quad\bepsilon\cdot(\btau_2\times\btau_3)=0.
\end{array}\right.
\end{equation}

\begin{equation}\label{Sep35}
\left\{\begin{array}{l}
\bsigma_3\times\bsigma_1=\bigl(U^1(k^1,k^3)\bmu^1+U^2(k^1,k^3)\bmu^2\bigr)\times\bkappa,\quad\bsigma_1=\bkappa;\\[2mm]
      \btau_2\cdot(\bmu^1\times\bkappa)=0,\quad\btau_2\cdot(\bmu^2\times\bkappa)=0,\quad\bkappa\cdot(\btau_2\times\btau_3)=1.
\end{array}\right.
\end{equation}
In these equations the sought vectors are $\bsigma(k^1,k^3)$, and $\btau(k^2,k^3)$; their lower indices $i$ denote derivatives with  respect to $k^i$. All the remaining vectors depend arbitrarily on variable $k^3$, and are linearly independent for every value of $k^3$. Arbitrary functions $U$, $U^1$, $U^2$, and $1$ are linearly independent as functions on variable $k^1$ over the whole domain of each system.
\end{theorem}
Proof of Theorem \ref{T1} can be found in \ref{AppA}. Each of the systems (\ref{Sep11})--(\ref{Sep35}) is a nonlinear overdetermined system of differential equations with sought functions $\bsigma$, and $\btau$. Each system possesses its own group of equivalence transformations. The group will be utilized to discard the inessential arbitrariness of arbitrary functions and vectors enumerated in the formulation of Theorem \ref{T1}.

In the following two sections the investigation of compatibility conditions and integration of each of obtained in Theorem \ref{T1} systems of equations is performed. The equivalence transformations of each system are used in order to bring its solution to a certain canonical form. The result of computations is formulated in Theorem \ref{T2}. The general standard form of the solution is given in Theorem \ref{T3}.

\section{Integration of equations for $s=0$ and $s=1$} The analysis is started from system (\ref{Sep11}). Integration on $k^1$ of the second equation of (\ref{Sep11}) gives
\[\bsigma=k^1\bnu(k^3)+\bdelta(k^3).\]
Substitution of this result into the first equation of (\ref{Sep11}) leads to
\[k^1\bnu'(k^3)\times\bnu(k^3)+\bdelta'(k^3)\times\bnu(k^3)=0.\]
As before, prime denotes the derivative with respect to the only argument $k^3$. Splitting with respect to the independent variable $k^1$ and integration of the subsequent equations produce
\[\bnu=\alpha(k^3)\bfeta,\quad \bdelta=\beta(k^3)\bnu,\quad\bfeta=\const.\]
with an arbitrary function $\alpha(k^3)$ and a constant vector $\bfeta$. By the action of transformations (\ref{SBMHD:BT}) and (\ref{SBMHD:IT}) it is always possible to fix constants as $\alpha=1$, $\beta=0$. Thus, vector $\bsigma$ has the form
\begin{equation}\label{SBMHD:20}
\bsigma=k^1\bfeta.
\end{equation}
The last of equations (\ref{Sep11}) after the substitution of $\sigma$ from (\ref{SBMHD:20}) can be written as follows:
\begin{equation}\label{SBMHD:21}
\bfeta\cdot\bigl(\btau_2\times\btau_3\bigr)=1.
\end{equation}
The obtained solution coincides with (\ref{SBMHD:s0}).

For $s=1$ equations (\ref{Sep21}) are considered. From the first of equations (\ref{Sep21}) it follows that vectors $\bsigma_3$ and $\bsigma_1$ are proportional which, in turn, imply
\[\bsigma=\bPhi(\lambda),\quad\lambda=\lambda(k^1,k^3),\quad \pd{\lambda}{k^1}\ne0.\]
Integration of the second of equations (\ref{Sep21}) on $k^1$ reduces the problem of determination of the vector field $\bsigma(k^1,k^3)$ to the solution of the following functional equation
\begin{equation}\label{SBMHD:29}
\bPhi\bigl(\lambda(k^1,k^3)\bigr)=u(k^1,k^3)\bepsilon(k^3)+k^1\bkappa(k^3)+\bdelta(k^3).
\end{equation}
Here $\partial u/\partial k^1=U$. It is required to find all possible functions $\lambda(k^1,k^3)$, and vector fields $\bPhi(\lambda)$, such that equation (\ref{SBMHD:29}) is satisfied identically on variables $k^1$ and $k^3$ for some functions $u(k^1,k^3)$, $\partial^2u/{\partial k^1}^2\ne0$ and vector fields $\bepsilon(k^3)$, $\bkappa(k^3)$, $\bdelta(k^3)$.

Differentiation of the equality (\ref{SBMHD:29}) with respect to $k^1$ yields
\[\bPhi'(\lambda)\lambda_1=u_1\bepsilon+\bkappa.\]
Here, as usual, the lower index denotes the derivative with respect to $k^1$.
Vector multiplication of this equations on $\bepsilon$, and subsequent dot multiplication on $\bkappa$ gives the following classifying equation:
\[\lambda_1(\bPhi'\times\bepsilon)\cdot\bkappa=0.\]
By definition, $\lambda_1\ne0$, hence the mixed product is equal to zero:
\[\bPhi'\cdot(\bepsilon\times\bkappa)=0.\]
Note, that vectors $\bPhi'$ and $\bepsilon\times\bkappa$ depend on different sets of variables, which allows application of a modification of Lemma \ref{SBMHD:l2} for the separation of variables in this equation. Due to the linear independence of vectors $\bepsilon$ and $\bkappa$ the following cases are possible:
\begin{enumerate}
\renewcommand{\theenumi}{(\alph{enumi})}
\item $\bPhi(\lambda)=f(\lambda)\be^1$, at that $\be^1\cdot(\bepsilon\times\bkappa)=0$;\label{SBMHD:c1}
\item $\bPhi(\lambda)=f(\lambda)\be^1+g(\lambda)\be^2$, at that $\be^1\cdot(\bepsilon\times\bkappa)=0$, $\be^2\cdot(\bepsilon\times\bkappa)=0$.\label{SBMHD:c2}
\end{enumerate}
Here $\be^1$ and $\be^2$ are orthonormal constant vectors.

In case \ref{SBMHD:c1} one has
\[\bkappa=v(k^3)\be^1+w(k^3)\bepsilon.\]
Substitution of obtained representations into equation (\ref{SBMHD:29}) yields
\[
f(\lambda)\be^1=\bigl(u(k^1,k^3)+k^1w(k^3)\bigr)\bepsilon(k^3)+k^1v(k^3)\be^1+\bdelta(k^3).
\]
Vector multiplication of this equation to $\be^1$ and splitting on $k^1$ gives $\bepsilon\times\be^1=0$, $\bdelta\times\be^1=0$. Hence, vectors $\bepsilon$ and $\bdelta$ are also proportional to vector $\be^1$. Substitution of vectors $\bkappa$ and $\bepsilon$ into the second of equations (\ref{Sep22}) leads to the contradiction.

Now, case \ref{SBMHD:c2} is observed. Here
\[\bkappa=v(k^3)\be^1+w(k^3)\be^2,\quad\bepsilon=m(k^3)\be^1+n(k^3)\be^2.\]
Form equation (\ref{SBMHD:29}) it follows that vector $\bdelta$ is in fact a linear combination of vectors $\be^1$ and $\be^2$. Thus, the equation follows:
\begin{equation}\label{SBMHD:31}
\bsigma=f(\lambda)\be^1+g(\lambda)\be^2.
\end{equation}
Function $\lambda(k^1,k^3)$ satisfies an implicit equation obtained by the elimination of function $u(k^1,k^3)$ from the following two relations:
\begin{equation}\label{SBMHD:32}
\begin{array}{l}
f(\lambda)=u(k^1,k^3)m(k^3)+k^1v(k^3)+\alpha(k^3),\\[2mm]
g(\lambda)=u(k^1,k^3)n(k^3)+k^1w(k^3)+\beta(k^3).
\end{array}
\end{equation}
Here $m$, $n$, $v$, $w$, $\alpha$, $\beta$ are arbitrary functions of $k^3$ such that
\begin{equation}\label{SBMHD:37}
\Delta=nv-mw\ne0.
\end{equation}

Vector $\bsigma$ determined by equation (\ref{SBMHD:31}), (\ref{SBMHD:32}) has the form
\[\bsigma_1=u(k^1,k^3)\bigl(m(k^3)\be^1+n(k^3)\be^2\bigr)+v(k^3)\be^1+w(k^3)\be^2.\]
Substitution into equation $\bsigma_1\cdot(\btau_2\times\btau_3)=1$ and splitting with respect to $k^1$ lead to the equation
\begin{equation}\label{SBMHD:35}
\begin{array}{l}
(m\be^1+n\be^2)\cdot(\btau_2\times\btau_3)=0,\\[2mm]
(v\be^1+w\be^2)\cdot(\btau_2\times\btau_3)=1.
\end{array}
\end{equation}
Vector $\btau$ has the following components in the orthonormal basis of vectors $\be^1$, $\be^2$ and $\be^3$:
\[\btau=(\tau^1,\tau^2,\tau^3)^{\mbox{T}}.\]
Equations (\ref{SBMHD:35}) can be written as follows:
\begin{equation}\label{SBMHD:36}
\left|
\begin{array}{ccc}
m&n&0\\
\tau^1_2&\tau^2_2&\tau^3_2\\
\tau^1_3&\tau^2_3&\tau^3_3
\end{array}
\right|=0,\quad
\left|
\begin{array}{ccc}
v&w&0\\
\tau^1_2&\tau^2_2&\tau^3_2\\
\tau^1_3&\tau^2_3&\tau^3_3
\end{array}
\right|=1
\end{equation}
As before, lower indices denote derivatives on variables $k^2$ and $k^3$. By virtue of inequality (\ref{SBMHD:37}), equations (\ref{SBMHD:36}) can be represented in the form of two linear equations for functions $\tau^1(k^2,k^3)$ and $\tau^2(k^2,k^3)$ provided function $\tau^3(k^2,k^3)$ is arbitrarily fixed:
\begin{equation}\label{SBMHD:38}
\left|
\begin{array}{cc}
\tau^1_2&\tau^3_2\\
\tau^1_3&\tau^3_3
\end{array}
\right|=\frac{m}{\Delta},\quad
\left|
\begin{array}{cc}
\tau^2_2&\tau^3_2\\
\tau^2_3&\tau^3_3
\end{array}
\right|=\frac{n}{\Delta},\quad\Delta=nv-mw.
\end{equation}
Finally, the solution takes the form
\begin{equation}\label{SBMHD:s2}
\bx=\bigl(f(\lambda)+\tau^1\bigr)\be^1+\bigl(g(\lambda)+\tau^2\bigr)\be^2+\tau^3\be^3,
\end{equation}
where $f(\lambda)$ and $g(\lambda)$ are arbitrary functions; $\lambda(k^1,k^3)$ is determined by the implicit equation that is obtained by the elimination of function $u(k^1,k^3)$ in relations (\ref{SBMHD:32}); $\tau^1(k^2,k^3)$ and $\tau^2(k^2,k^3)$ are solutions of linear equations (\ref{SBMHD:38}) with arbitrarily fixed function $\tau^3(k^2,k^3)$ and arbitrary functions $m(k^3)$, $n(k^3)$, $v(k^3)$, $w(k^3)$, taken form (\ref{SBMHD:32}).
This solution is simplified below by the action of equivalence transformations.

Let $m=0$. By using the equivalence transformations (\ref{SBMHD:BT}) and (\ref{SBMHD:IT}) one can make $v=1$, $\alpha=0$ in the first of equations (\ref{SBMHD:32}). According to the obtained equation, by virtue of the arbitrariness of function $f$, one can assume $\lambda=k^1$. At that, the second of equations (\ref{SBMHD:32}) does not lead to additional restrictions to functions $g$ and $\lambda$ due to the arbitrariness of function $u$. One can assume the right-hand side of equation (\ref{SBMHD:32}) to be a definition of a new function $u(k^1,k^3)$, i.e. $n=1$, $w=\beta=0$. At that, from the first equation of (\ref{SBMHD:38}) it follows that functions $\tau^1$ and $\tau^3$ are functionally dependent. The solution acquires the following canonical form:
\begin{equation}\label{SBMHD:s4}
\fl\bx=\left(k^1+F\bigl(\tau^3(k^1,k^3)\bigr)\right)\be^{1}+\bigl(g(k^1)+\tau^2(k^2,k^3)\bigr)\be^{2}+\tau^3(k^2,k^3)\be^{3},\quad
\left|
\begin{array}{cc}
\tau^2_2&\tau^3_2\\
\tau^2_3&\tau^3_3
\end{array}
\right|=1.
\end{equation}
Here $F$, and $g$ are arbitrary functions.

In case $m\neq0$ by using the arbitrariness of function $u(k^1,k^3)$, the right-hand side of the first of equations (\ref{SBMHD:32}) can be taken as a new function $u$. Then $m=1$, $v=\alpha=0$. Functions $n$, $w$ and $\beta$ in the second of equations (\ref{SBMHD:32}) are changed correspondingly. Next, by using the equivalence transformations (\ref{SBMHD:BT}) and (\ref{SBMHD:IT}) one makes $w=1$, $\beta=0$. Thus,
\numparts
\begin{eqnarray}\label{SBMHD:d61}
&&g(\lambda)=f(\lambda)n(k^3)+k^1,\\[3mm]
&&\displaystyle \left|
\begin{array}{cc}
\tau^1_2&\tau^3_2\\
\tau^1_3&\tau^3_3
\end{array}
\right|=-1,\quad
\left|
\begin{array}{cc}
\tau^2_2&\tau^3_2\\
\tau^2_3&\tau^3_3
\end{array}
\right|=-n(k^3).\label{SBMHD:d62}
\end{eqnarray}
\endnumparts
Finally, the solution takes the canonical form as
\begin{equation}\label{SBMHD:s41}
\bx=(f(\lambda)+\tau^1(k^2,k^3))\be^{1}+(g(\lambda)+\tau^2(k^2,k^3))\be^{2}+\tau^3(k^2,k^3)\be^{3}.
\end{equation}
Here the relation (\ref{SBMHD:d61}) for arbitrary choice of functions $f$, $g$ and $n$ determines function $\lambda(k^1,k^3)$. Functions $\tau^1$, $\tau^2$, and $\tau^3$ should satisfy equations (\ref{SBMHD:d62}).

Hence, equations (\ref{Sep21}) are completely integrated. Solutions are given by formulae (\ref{SBMHD:s4}) and (\ref{SBMHD:s41}).

For $s=1$ equations (\ref{Sep22}) are considered. Integration of the second of the equations (\ref{Sep22}) leads to the following representation of vector $\bsigma$:
\begin{equation}\label{SBMHD:12a1}
\bsigma=u(k^1,k^3)\bepsilon(k^3)+\bdelta(k^3).
\end{equation}
Here $U=\partial u/\partial k^1$. Differentiation of equation (\ref{SBMHD:12a1}) on variables $k^1$, $k^3$ and substitution into the first of equations (\ref{Sep22}) lead to the relation
\begin{equation}\label{SBMHD:12a2}
\bigl(uu_1\,\bepsilon'+u_1\,\bdelta'-\bnu\bigr)\times\bepsilon=0.
\end{equation}
In order to separate variables in equation (\ref{SBMHD:12a2}), cases of different space of the dimension of the linear space spanned by the 3D vector $(uu_1,u_1,1)$ for various values of  $k^1$ should be investigated. By virtue of condition $u_{11}\ne0$ the last two components of the vector are always linearly independent, hence the dimension can be either 2 or 3.

i) $\dim\{(uu_1,u_1,1)\}=3$. In this case all components are linearly independent, which allows the separation of variables in equation (\ref{SBMHD:12a2}):
\[\bepsilon'\times\bepsilon=0,\quad\bdelta'\times\bepsilon=0,\quad\bnu\times\bepsilon=0.\]
By the integration one finds
\begin{equation}\label{SBMHD:12a13}
\bepsilon=\varepsilon(k^3)\bfeta,\quad\bdelta=\delta(k^3)\bfeta,\quad\bnu=\nu(k^3)\bfeta,\quad\bfeta=\const.
\end{equation}
Substitution of representations (\ref{SBMHD:12a13}) into the second series of equations (\ref{Sep22}) leads to the contradiction.

ii) $\dim\{(uu_1,u_1,1)\}=2$. In this case there exists a linear dependence between components of vector $(uu_1,u_1,1)$. Due to the linear independence of components $u_1$ and $1$ the linear relation can be chosen as
\[uu_1+a(k^3) u_1+b(k^3)=0.\]
By integration one finds
\[\frac{1}{2}\,u^2+a(k^3) u+b(k^3) k^1=c(k^3).\]
Function $u$ is determined accurate to an arbitrary additive that depends on $k^3$. This arbitrariness allows one to zero function $a$. Next, the equivalence transformations (\ref{SBMHD:BT}) and (\ref{SBMHD:IT}) allows to normalize parameters as $b=-1/2$, $c=0$. Thus, in this case function $u$ is equivalent to the following one:
\begin{equation}\label{SBMHD:12a14}
u=\pm\sqrt{|k^1|}.
\end{equation}
Separation of variables in equation (\ref{SBMHD:12a2}) leads to the relation
\[\left(\frac{1}{2}\,\bepsilon'-\bnu\right)\times\bepsilon=0,\quad\bdelta'\times\bepsilon=0.\]
From this, accurate to insufficient additive, one obtains
\begin{equation}\label{SBMHD:12a15}
\bnu=\frac{1}{2}\,\bepsilon',\quad \bdelta'=\alpha(k^3)\bepsilon.
\end{equation}
Under this condition the second series of equations (\ref{Sep22}) reads
\begin{equation}\label{SBMHD:12a16}
\bepsilon'\cdot\bigl(\bepsilon\times\btau_2\bigr)=2,\quad\btau_3\cdot\bigl(\bepsilon\times\btau_2\bigr)=0.
\end{equation}
By virtue of the first of equations (\ref{SBMHD:12a16}) vectors $\bepsilon$ and $\btau_2$ are linearly independent, whereas due to the second of this equations vector $\btau_3$ can be linearly expressed in terms of $\btau_2$ and $\bepsilon$. This fixes vector $\btau$ in the form
\begin{equation}\label{SBMHD:12a17}
\btau=\bPhi(\lambda)+\bomega(k^3),\quad\lambda=k^2+a(k^3),\quad\bomega'=\beta(k^3)\bepsilon
\end{equation}
with arbitrary vector $\bPhi$ and arbitrary functions $a$ and $\beta$. From the first of equations (\ref{SBMHD:12a16}) it follows the relation
\begin{equation}\label{SBMHD:12a18}
\bepsilon'\cdot\bigl(\bepsilon\times\bPhi'(\lambda)\bigr)=2.
\end{equation}
Independence of vector $\bepsilon$ on variable $k^2$ implies
\[\bPhi=\lambda\bfeta+\bzeta,\quad\bfeta,\;\bzeta=\const.\]
Substitution into equation (\ref{SBMHD:12a18}) yields
\begin{equation}\label{SBMHD:12a19}
\bepsilon'\cdot\bigl(\bepsilon\times\bfeta\bigr)=2.
\end{equation}
Differentiation of equation (\ref{SBMHD:12a19}) on $k^3$ gives
\[\bepsilon''\cdot\bigl(\bepsilon\times\bfeta\bigr)=0\]
this implies that vector $\bepsilon''$ can be expressed linearly in terms of $\bepsilon$ and $\bfeta$. Hence,
\begin{equation}\label{SBMHD:12a110}
\bepsilon=\rme^{\gamma k^3}\bfeta^1+\rme^{-\gamma k^3}\bfeta^2+C\bfeta
\end{equation}
with arbitrary vectors $\gamma$, $C$ and arbitrary constant vectors $\bfeta^1$ and $\bfeta^2$. Substitution of representation (\ref{SBMHD:12a110}) into equation (\ref{SBMHD:12a19}) gives a restriction on vectors $\bfeta^1$ and $\bfeta^2$:
\begin{equation}\label{SBMHD:12a111}
\gamma\bfeta\cdot\bigl(\bfeta^1\times\bfeta^2\bigr)=1.
\end{equation}
Finally, accurate to insufficient terms, one obtains the solution
\begin{equation}\label{SBMHD:s11}
\begin{array}{l}
\displaystyle\bx=\pm\sqrt{|k^1|}\left(\rme^{\gamma k^3}\bfeta^1+\rme^{-\gamma k^3}\bfeta^2+C\bfeta\right)+k^2\bfeta+\\[2mm]
\displaystyle\qquad+\int\alpha(k^3)\rme^{\gamma k^3}\rmd k^3\bfeta^1+
\int\alpha(k^3)\rme^{-\gamma k^3}\rmd k^3\bfeta^2+C\int\alpha(k^3)\rmd k^3\bfeta.
\end{array}
\end{equation}
Here $\gamma$ and $C$ are arbitrary constant such that $\gamma$ may take complex values; in this case vectors $\eta^i$ should be chosen such that the whole solution is real. Function $\alpha(k^3)$ is arbitrary; constant vectors $\bfeta$, $\bfeta^1$, and $\bfeta^2$ satisfy equations (\ref{SBMHD:12a111}). Thus, the canonical form of the solution of equations (\ref{Sep23}) in the observed case is given by equation (\ref{SBMHD:s11}).

Finally, for $s=1$ equations (\ref{Sep23}) are observed. By virtue of the second of equations (\ref{Sep23}) vector $\bsigma$ depends linearly on $k^1$ and, hence, the solution is reduced to either of solutions (\ref{SBMHD:s0}) or (\ref{SBMHD:slin}). This remark completes the investigation of case $s=1$.

\section{Integration of equations in case $s=2$}
Equations (\ref{Sep31}) are now observed. By virtue of the arbitrariness in the choice of linearly independent functions $U^1$ and $U^2$ this system admits the following equivalence transformations:
\begin{eqnarray}\label{SBMHD:47}
&&\begin{array}{cc}
\tbepsilon^1=A^1(k^3)\bepsilon^1+A^2(k^3)\bepsilon^2,\\[2mm]
\tbepsilon^2=B^1(k^3)\bepsilon^1+B^2(k^3)\bepsilon^2,
\end{array}\quad A^1B^2\ne A^2B^1,
\\[1mm]\nonumber
&&\begin{array}{l}
\tbkappa=\bkappa+\omega^1(k^3)\bepsilon^1+\omega^2(k^3)\bepsilon^2.
\end{array}
\end{eqnarray}
By the action of transformations (\ref{SBMHD:47}) one can bring vectors $\bepsilon^{1}$, $\bepsilon^{2}$ to the following form:
\begin{equation}\label{SBMHD:48}
\bepsilon^{1}=\be^{1}+\alpha(k^{3})\be^{3},\;\;
\bepsilon^{2}=\be^{2}+\beta(k^{3})\be^{3},\;\;\bkappa=\be^{3},
\end{equation}
where $\alpha$, $\beta$ are arbitrary functions.

Integration of the second equation of  (\ref{Sep31}) on $k^1$ leads to the following representation for vector $\bsigma$:
\begin{equation}\label{SBMHD:43}
\bsigma=u^1(k^1,k^3)\bepsilon^1(k^3)+u^2(k^1,k^3)\bepsilon^2(k^3)+k^1\bkappa(k^3).
\end{equation}
Here $U^i=\partial u^i/\partial k^1$. An arbitrary additive vector depending on $k^3$ that appears as a ``constant of integration" can be zeroed by the suitable change of functions $u^1$ and $u^2$, and by application of transformation (\ref{SBMHD:IT}). From representations (\ref{SBMHD:48}), (\ref{SBMHD:43}) for vector $\bsigma$ one finds
\[\bsigma=u^1(k^1,k^3)\be^{1}+u^2(k^1,k^3)\be^{2}+S(k^1,k^3)\be^{3},\quad S=\alpha u^1+\beta u^2+k^1.\]
This implies
\begin{equation}\label{SBMHD:52}
\bsigma_1=u^1_1\be^{1}+u^2_1\be^{2}+S_1\be^{3}, \quad \bsigma_3=u^1_3\be^{1}+u^2_3\be^{2}+S_3\be^{3}.
\end{equation}
Hence,
\[\bsigma_3\times\bsigma_1=\bigl(u^1_3u^2_1-u^1_1u^2_3\bigr)\be^{3}-\bigl(u^1_3S_1-u^1_1S_3\bigr)\be^{2}+ \bigl(u^2_3S_1-u^2_1S_3\bigr)\be^{1}=0\]
The latter equation is satisfied only if
\begin{equation}\label{SBMHD:49}
u^1_3u^2_1=u^1_1u^2_3, \quad u^1_3S_1=u^1_1S_3,\quad u^2_3S_1=u^2_1S_3.
\end{equation}
Equations (\ref{SBMHD:49}) imply the functional dependence
\begin{equation}\label{SBMHD:50}
u^1=u^1\bigl(\lambda(k^1,k^3)\bigr),\quad u^2=u^2\bigl(\lambda(k^1,k^3)\bigr),\quad S=S\bigl(\lambda(k^1,k^3)\bigr).
\end{equation}
Function $\lambda(k^1,k^3)$ can be found from the following equation:
\begin{equation}\label{SBMHD:51}
\alpha u^1(\lambda)+\beta u^2(\lambda)+k^1=S(\lambda)
\end{equation}
For some functions $\alpha(k^3)$, $\beta(k^3)$, $u^1(\lambda)$, $u^2(\lambda)$ and $S(\lambda)$.

Finally, vector $\bsigma$ has the form defined by the equation (\ref{SBMHD:43}) with functions $u^1$, $u^2$, given by relations (\ref{SBMHD:50}), (\ref{SBMHD:51}) and vectors $\bepsilon^1$, $\bepsilon^2$, $\bkappa$, satisfying equations (\ref{SBMHD:48}) with arbitrary functions $\alpha$, $\beta$.

Equations for vector $\btau(k^2,k^3)$ remain. After substitution of obtained representation for vector $\bsigma$ in equation (\ref{SBMHD:17}) and splitting with respect to linearly independent functions $u^1$ and $u^2$ one obtains the following equations:
\begin{equation}\label{SBMHD:52a}
(\be^{1}+\alpha\be^{3})\cdot(\btau_2\times\btau_3)=0,\quad (\be^{2}+\beta\be^{3})\cdot(\btau_2\times\btau_3)=0, \quad\be^{3}\cdot(\btau_2\times\btau_3)=1.
\end{equation}

Vector $\btau$ can be written in the basis of linearly independent vectors $\bepsilon^1$, $\bepsilon^2$, $\bkappa$ as follows:
\begin{equation}\label{SBMHD:53}
\btau=\tau^1(\be^{1}+\alpha\be^{3})+\tau^2(\be^{2}+\beta\be^{3})+\tau^3\be^{3}.
\end{equation}
Here functions $\tau^i$ depends on $k^2$ and $k^3$. Substitution of representation (\ref{SBMHD:53}) in equation (\ref{SBMHD:52a}) gives the following relations:
\begin{eqnarray}\label{SBMHD:54a}
&&\left|\begin{array}{ll}
\tau^2_2 & \tau^3_2\\[1mm]
\tau^2_3 & \tau^3_3+\tau^1\alpha'+\tau^2\beta'
\end{array}\right|=0,\\[4mm]\label{SBMHD:54b}
&&\left|\begin{array}{ll}
\tau^1_2 & \tau^3_2\\[1mm]
\tau^1_3 & \tau^3_3+\tau^1\alpha'+\tau^2\beta'
\end{array}\right|=0,\\[4mm]\label{SBMHD:54c}
&&\left|\begin{array}{ll}
\tau^1_2 & \tau^2_2\\[1mm]
\tau^1_3 & \tau^2_3
\end{array}\right|=1.
\end{eqnarray}
By virtue of equation (\ref{SBMHD:54c}), equations (\ref{SBMHD:54a}) and (\ref{SBMHD:54b}) have non-trivial solutions only if
\begin{equation}\label{SBMHD:55}
\tau^3_2=0,\quad\tau^3_3+\tau^1\alpha'+\tau^2\beta'=0.
\end{equation}
As a consequence of equation (\ref{SBMHD:55}), function $\tau^3$ in representation (\ref{SBMHD:53}) depends only on $k^3$. For simplicity the notation $\tau^3=\gamma(k^3)$ is adopted. In the second equation of (\ref{SBMHD:55}) one distinguishes the following cases
\begin{enumerate}
\item $\alpha'=\beta'=0$;
\item $(\alpha')^2+(\beta')^2\ne0$;
\end{enumerate}
In the first case $\gamma'=0$, and functions $\tau^2$, $\tau^3$ are such that equation (\ref{SBMHD:54c}) is satisfied. This leads to a solution
\begin{equation}\label{solve:21}
\bx=\bigl(u^{1}(k^{1})+\tau^{1}(k^2,k^{3})\bigr)\be^1+\bigl(u^{2}(k^{1})+\tau^{2}(k^2,k^{3})\bigr)\be^2+k^1\be^{3},
\end{equation}
where $u^{1}$, $u^{2}$ are arbitrary functions; $\tau^1$, $\tau^2$ are arbitrary functions satisfying equation
\[
\left|\begin{array}{ll}
\tau^1_2    &  \tau^2_2\\[2mm]
\tau^1_3    & \tau^2_3
\end{array}\right|=1.
\]

In the second case by virtue of equation (\ref{SBMHD:55}) one obtains the representation
\begin{equation}\label{taurep}
\fl\tau^1=\beta'\varphi(k^2,k^3)-\frac{\alpha'\gamma'}{(\alpha')^2+(\beta')^2},\quad
\tau^2=-\alpha'\varphi(k^2,k^3)-\frac{\beta'\gamma'}{(\alpha')^2+(\beta')^2}.
\end{equation}
Substitution of this representation into equation (\ref{SBMHD:54c}) gives the following equation for function $\varphi$:
\[\fl\varphi\,\pd{\varphi}{k^2}(\alpha'\beta''-\beta'\alpha'')-
\pd{\varphi}{k^2}\left(\alpha'\left(\frac{\alpha'\gamma'}{(\alpha')^2+(\beta')^2}\right)'
+\beta'\left(\frac{\beta'\gamma'}{(\alpha')^2+(\beta')^2}\right)'\right)=1.\]
Integration of this representation on $k^2$ and simplification of the second term leads to the quadratic equation for function $\varphi$:
\begin{equation}\label{phieq}
\frac{\alpha'\beta''-\beta'\alpha''}{2}\,\varphi^2-\left(\gamma ''-\frac{\gamma '\left(\alpha '\alpha ''+\beta '\beta''\right)}{(\alpha ')^2+(\beta ')^2}\right)\varphi=k^2.
\end{equation}

The final form of the solution is
\begin{equation}\label{SBMHD:s3}
\bx=(u^1(\lambda)+\tau^1)\be^{1}+(u^2(\lambda)+\tau^2)\be^{2}+\bigl(S(\lambda)+\alpha\tau^1+\beta\tau^2+\gamma\bigr)\be^{3}.
\end{equation}
Here $\alpha(k^3)$, $\beta(k^3)$, $\gamma(k^3)$ are arbitrary functions. Functions $u^1(\lambda)$, $u^2(\lambda)$, $S(\lambda)$ are also arbitrary and linearly independent; function $\lambda(k^1,k^3)$ is given by the implicit equation (\ref{SBMHD:51}). Functions $\tau^1(k^2,k^3)$ and $\tau^2(k^2,k^3)$ is specified by relation (\ref{taurep}), in which function $\varphi$ is determined by the condition (\ref{phieq}).

For $s=2$ in the case (\ref{Sep32}) the system of equation admits the following equivalence transformations:
\begin{equation}\label{SBMHD:d4}
\begin{array}{c}
\overline{\bepsilon^{1}}=A^{1}(k^{3})\bepsilon^{1}+A^{2}(k^{3})\bepsilon^{2},\\[2mm]
\overline{\bepsilon^{2}}=B^{1}(k^{3})\bepsilon^{1}+B^{2}(k^{3})\bepsilon^{2},
\end{array}\quad A^1B^2\ne A^2B^1.
\end{equation}
In the orthonormal basis of vectors $\be^{1}$, $\be^{2}$, $\be^{3}$, vectors $\bepsilon^{1}$, $\bepsilon^{2}$ and $\btau$ have the following components:
\[
\bepsilon^{1}=(\varepsilon^{11},\varepsilon^{12},\varepsilon^{13})^{T},\quad
\bepsilon^{2}=(\varepsilon^{21},\varepsilon^{22},\varepsilon^{23})^{T},\quad
\btau=(\tau^{1},\tau^{2},\tau^{3})^{T}.
\]
By the action of  equivalence transformations (\ref{SBMHD:d4}) vectors $\bepsilon^{1}$ and $\bepsilon^{2}$ can be brought to a form
\begin{equation}\label{SBMHD:ss5}
\bepsilon^{1}=\be^{1}+\alpha(k^{3})\be^{3},\quad
\bepsilon^{2}=\be^{2}+\beta(k^{3})\be^{3}
\end{equation}
where $\alpha$ and $\beta$ are some arbitrary functions.

First, equations for vector $\bsigma$ are observed. Substitution of relation (\ref{SBMHD:ss5}) into the second of equations (\ref{Sep32}) gives
\begin{equation}\label{SBMHD:s9}
\bsigma_{1}=u^{1}_{1}\be^{1}+u^{2}_{1}\be^{2}+(\alpha u_{1}^{1}+\beta u^{2}_{1})\be^{3}
\end{equation}
Integration of this equation on $k^{1}$ and differentiation on $k^{3}$ produces
\begin{equation}\label{SBMHD:s10}
\bsigma_{3}=u^{1}_{3}\be^{1}+u^{2}_{3}\be^{2}+(\alpha' u^{1}+\alpha u^{1}_{3}+\beta' u^{2}+\beta u^{2}_{3}+\gamma')\be^{3}
\end{equation}
Vector product of these formulas, collection of coefficients at basic vectors, and equating of corresponding coefficients in the right-hand side leads to the following system of equations:
\begin{equation}\label{SBMHD:s111}
\fl u^{2}_{1}u_{3}^{1}-u_{1}^{1}u_{3}^{2}=1,\quad
u_{1}^{1}(\alpha'u^{1}+\beta'u^{2}+\gamma')=0,\quad
u^{2}_{1}(\alpha'u^{1}+\beta'u^{2}+\gamma')=0.
\end{equation}
Two cases are possible. Case a) $u_{1}^{1}=u_{1}^{2}=0$ contradicts to the linear independence of functions $U^1=u_{1}^{1}$ and $U^2=u_{1}^{2}$. In case b) $\alpha'u^{1}+\beta'u^{2}+\gamma'=0$ by virtue of linear independence of functions $U^{1}$ and $U^{2}$ one obtains $\alpha'(k^{3})=\beta'(k^{3})=\gamma'(k^{3})=0$. Then, functions $\alpha(k^{3})$, $\beta(k^{3})$ and $\gamma(k^{3})$ are constants, that can be taken zero without limiting the generality.

Equations for $\btau$ are now observed. Substitution of representations (\ref{SBMHD:ss5}) into equation for $\btau$ gives the following relations
\begin{eqnarray}\label{SBMHD:s5}
\begin{array}{cc}\label{SBMHD:s6}
\gamma(k^{3})\tau_{2}^{3}=1,
\end{array}
\\[4mm]
\left|\begin{array}{cc}\label{SBMHD:s7}
\tau_{2}^{2}&\tau^{3}_{2}\\
\tau_{2}^{3}&\tau^{3}_{3}
\end{array}
\right|=0,\quad
\left|\begin{array}{cc}
\tau_{2}^{1}&\tau^{1}_{3}\\
\tau_{2}^{3}&\tau^{3}_{3}
\end{array}
\right|=0.
\end{eqnarray}
Equations (\ref{SBMHD:s7}) imply
\[
\btau=\btau\bigl(\lambda(k^2,k^3)\bigr).
\]
By integration of equations (\ref{SBMHD:s6}) on $k^{2}$ one obtains
\[
\gamma(k^{3})\tau^{3}=k^{2}
\]
The arbitrary constant of integration appearing at the integration can be zeroed by the equivalence transformation (\ref{SBMHD:IT}). By applying transformation (\ref{SBMHD:BT}) one can make $\gamma(k^{3})=1$. By virtue of the equation for function $\tau^3$ one can assume $\lambda=k^2$.

The solution takes the following form:
\[
\bx=\bigl(u^1(k^1,k^3)+\tau^1(k^2)\bigr)\be^{1}+\bigl(u^2(k^1,k^3)+\tau^2(k^2)\bigr)\be^{2}+k^2\be^{3},
\]
where functions $u^1(k^1,k^3)$, $u^2(k^1,k^3)$ satisfy the relation $u^{2}_{1}u_{3}^{1}-u_{1}^{1}u_{3}^{2}=1$; function $\tau^1(k^2)$, $\tau^2(k^2)$ are arbitrary. This solution coincide with solution (\ref{solve:21}).

In case of equations (\ref{Sep33}) the following equivalence transformations are admitted:
\[
\overline{\bkappa}=A(k^{3})\bkappa+B(k^{3})\bepsilon,\quad\overline{\bepsilon}=C(k^{3})\bepsilon,\quad A,C\ne0.
\]
Under this transformation vectors $\bepsilon$, $\bkappa$ and $\btau$ can be represented in the orthonormal basis of vectors $\be^{i}$ as
\begin{eqnarray}\label{SBMHD:s12}
%\begin{array}{c}
\bepsilon=\be^{1}+\alpha^{2}(k^{3})\be^{2}+\alpha^{3}(k^{3})\be^{3},\quad
\bkappa=\be^{2}+\beta\be^{3},\quad
\btau=(\tau^{1},\tau^{2},\tau^{3})^{T}.
%\end{array}
\end{eqnarray}

In the same manner equations for vector $\bsigma$ give
\begin{equation}\label{f:2}
\fl\begin{array}{l}
u^{1}_{1}u^{1}(\alpha^{2})'-u_{3}^{1}=u^{2},\\[2mm]
u_{1}^{1}u^{1}(\alpha^{3})'+u_{1}^{1}k^{1}\beta'-u_{3}^{1}\beta=u^{2}\beta,\\[2mm]
u_{3}^{1}\alpha^{3}+(u_{1}^{1}\alpha^{2}+1)(u^{1}(\alpha^{3})'-k^{1}\beta')- u^{1}_{3}\alpha^{2}\beta-u^{1}(\alpha^{2})'(u_{1}^{1}\alpha^{3}+\beta)=(\alpha^{2}\beta-\alpha^{3})u^{2}.
\end{array}
\end{equation}
Two former equations of this system imply
\[
u_{1}^{1}\bigl(u^{1}((\alpha^{2})'\beta-(\alpha^{3})')-k^{1}\beta'\bigr)=0.
\]
This equation is satisfied either if $u^{1}_{1}=0$, or if $u^{1}\bigl((\alpha^{2})'\beta-(\alpha^{3})'\bigr)-k^{1}\beta'=0$. First case contradicts to the linear independence of functions $U^1$ and $U^2$. The second case leads to equations
\[
(\alpha^{2})'\beta-(\alpha^{3})'=0,\;\;\beta'=0.
\]
Hence, $\beta=\beta_{0}=\const$, $\alpha^{2}\beta_{0}-\alpha^{3}=m=\const$. The third of equations (\ref{f:2}) is satisfied identically.

Now equations for vector $\btau$ are observed. Substitution of representations (\ref{SBMHD:s12}) into equations (\ref{Sep33}) gives
\begin{eqnarray}\label{SBMHD:ss1}
\begin{array}{cc}
\tau_{2}^{1}m-\tau^{2}_{2}\beta_{0}+\tau^{3}_{2}=0,\label{SBMHD:ss2}
\end{array}
\\[4mm]
\left|\begin{array}{cc}\label{SBMHD:ss3}
\tau_{2}^{2}&\tau^{3}_{2}\\
\tau_{2}^{3}&\tau^{3}_{3}
\end{array}
\right|-\alpha^{2}\left|
\begin{array}{cc}
\tau_{2}^{1}&\tau^{3}_{2}\\
\tau_{3}^{1}&\tau^{3}_{3}
\end{array}
\right|+\alpha^{3}\left|
\begin{array}{cc}
\tau_{2}^{1}&\tau^{2}_{2}\\
\tau_{3}^{1}&\tau^{2}_{3}
\end{array}
\right|=0,
\\[4mm]\label{SBMHD:ss4}
-\left|\begin{array}{cc}
\tau_{2}^{1}&\tau^{3}_{2}\\
\tau_{3}^{1}&\tau^{3}_{3}
\end{array}
\right|+\beta\left|
\begin{array}{cc}
\tau_{2}^{1}&\tau^{2}_{2}\\
\tau_{3}^{1}&\tau^{2}_{3}
\end{array}
\right|=1.
\end{eqnarray}
Form the first equation of the system it follows
\[
\tau^{3}=\delta(k^{3})+\beta_{0}\tau^{2}-m\tau^{1}.
\]
Substitution of $\tau^{3}$ into equation (\ref{SBMHD:ss4}) yields
\[
\tau^{1}=-\frac{k^{2}}{\delta'}+\nu(k^{3}).
\]
By plugging functions $\tau^{3}$, $\tau^{2}$ into equation (\ref{SBMHD:ss3}) one obtains
\[
\tau^{2}=\alpha^{2}\left(-\frac{k^{2}}{\delta'}+\nu\right)+\gamma(k^{3}).
\]
The final form of the solution is the following:
\begin{equation}\label{Sol35}
\fl\begin{array}{l}
\displaystyle\bx=\left(u^{1}+\nu-\frac{k^{2}}{\delta'}\right)\be^{1}+\left(k^{1}+u^{1}\alpha^{2}+
\gamma+\alpha^{2}\nu-\frac{k^{2}\alpha^{2}}{\delta'}\right)\be^{2}+\\[2mm]
\displaystyle\qquad+\left(k^{1}\beta_{0}+u^{1}(m+\beta_{0}\alpha^{2})+\beta_{0}\gamma+\delta+m\nu+
\beta_{0}\alpha^{2}\nu-\frac{k^{2}(m+\beta_{0}\alpha^{2})}{\delta'}\right)\be^{3}
\end{array}
\end{equation}
Here $u^1(k^1,k^3)$, $n(k^3)$, $\delta(k^3)$, $\alpha^2(k^3)$, $\gamma(k^3)$ are arbitrary functions; $m$, $\beta_0$ are arbitrary constants.

Now equations (\ref{Sep34}) are investigated. This system admits the following equivalence transformations:
\begin{equation}\label{SBMHD:sss5}
\begin{array}{l}
\overline{\bepsilon}=A(k^{3})\bepsilon,\quad\overline{\bnu}=B(k^3)\bnu+C(k^{3})\bmu+D(k^{3})\bepsilon,\\[2mm]
\overline{\bmu}=E(k^{3})\bmu+F(k^{3})\bepsilon,\quad A,B,E\ne0.
\end{array}
\end{equation}
Under the action of transformations (\ref{SBMHD:sss5}) vectors $\bepsilon$, $\bmu$ and $\bnu$ can be represented in the orthonormal basis $\be^{i}$ as follows
\begin{equation}\label{SBMHD:sss6}
%\begin{array}{c}
\bepsilon=\be^{1}+\alpha^{2}(k^{3})\be^{2}+\alpha^{3}(k^{3})\be^{3},\quad
\bmu=\be^{2}+\beta(k^3)\be^{3},\quad
\bnu=\be^{3}.
%\end{array}
\end{equation}

Let us observe equations for $\bsigma$. Substitution of representations $(\ref{SBMHD:sss6})$ into equation (\ref{Sep34}) and comparison of respective components gives the following equations
\[
\begin{array}{l}
\alpha^{3}\bigl(-u^{2}_{1}+u^{1}(\alpha^{2})'u_{1}^{1}\bigr)+\alpha^{2}\bigl(u^{2}_{1}\beta-u^{1}(\alpha^{3})'+1\bigr)=0,\\[2mm]
-u_{1}^{2}\beta+u^{1}(\alpha^{3})'u^{1}_{1}=1,\\[2mm]
u^{2}_{1}=u^{1}(\alpha^{2})'u^{1}_{1}.
\end{array}
\]
The first of equations is satisfied identically by virtue of the remaining two. Substitution of $u^{2}_{1}$ into the second equation of the system and integration with respect to $k^{1}$ accurate to insufficient constant of integration gives an equation for function $u^{1}$:
\[
(u^{1})^{2}\bigl(-\beta(\alpha^{2})'+(\alpha^{3})'\bigr)=2k^{1}.
\]
Hence,
\[
u^{1}=\pm\sqrt{\frac{2k^{1}}{(\alpha^{3})'-\beta(\alpha^{2})'}}
\]
Vector $\bsigma$ takes the form
\[
\displaystyle \bsigma=\pm\sqrt{\frac{2k^{1}}{(\alpha^{3})'-\beta(\alpha^{2})'}}\left(\be^{1}+\alpha^{2}(k^{3})\be^{2}+\alpha^{3}(k^{3})\be^{3}\right).
\]

Now we solve equations for vector $\btau$. Substitution of equations (\ref{SBMHD:sss6}) into (\ref{Sep34}) gives the following system of equations:
\begin{eqnarray}\label{SBMHD:sss8}
(\alpha^{3}-\alpha^{2}\beta)\tau_{2}^{1}+\beta\tau_{2}^{2}-\tau_{2}^{3}=0,\quad-\alpha^{2}\tau_{2}^{1}+\tau_{2}^{2}=1,
\\[4mm]
\left|\begin{array}{cc}
\tau^{2}_{2}&\tau^{3}_{2}\\
\tau^{2}_{3}&\tau^{3}_{3}
\end{array}
\right|-\alpha^{2}\left|
\begin{array}{cc}
\tau_{2}^{1}&\tau^{3}_{2}\\
\tau_{3}^{1}&\tau^{3}_{3}
\end{array}
\right|+\alpha^{3}\left|
\begin{array}{cc}
\tau_{2}^{1}&\tau^{2}_{2}\\
\tau_{3}^{1}&\tau^{2}_{3}
\end{array}
\right|=0.
\end{eqnarray}
Form the second equation of system (\ref{SBMHD:sss8}) by integration on $k^{2}$ accurate to an insufficient function one obtains
\[\tau^{2}=k^{2}+\alpha^{2}\tau^{1}.\]
Substitution of $\tau^{2}$ into the first equation of (\ref{SBMHD:sss8}) and integration on $k^{2}$ give
\[\tau^{3}=\beta k^{2}+\alpha^{3}\tau^{1}+\gamma.\]
Substitution of $\tau^{3}$ and $\tau^{2}$ into the first equation of the system and integration yield
\[
\tau^{1}=-\frac{k^{2}\beta'+\gamma'}{(\alpha^{3})'-\beta(\alpha^{2})'}.
\]
Finally, we obtain the solution in the form:
\begin{equation}\label{sol351}
\begin{array}{l}
\displaystyle\bx=\left(\pm\sqrt{\frac{2k^{1}}{(\alpha^{3})'-\beta(\alpha^{2})'}}-\frac{k^{2}\beta'+
\gamma'}{(\alpha^{3})'-\beta(\alpha^{2})'}\right)\be^{1}+\\[4mm]
\displaystyle\qquad+ \left(k^{2}+\alpha^{2}\left(\pm\sqrt{\frac{2k^{1}}{(\alpha^{3})'-\beta(\alpha^{2})'}}-
\frac{k^{2}\beta'+\gamma'}{(\alpha^{3})'-\beta(\alpha^{2})'}\right)\right)\be^{2}+\\[4mm]
\displaystyle\qquad\qquad+ \left(k^{2}\beta+\gamma+\alpha^{3}\left(\pm\sqrt{\frac{2k^{1}}{(\alpha^{3})'-\beta(\alpha^{2})'}}-
\frac{k^{2}\beta'+\gamma'}{(\alpha^{3})'-\beta(\alpha^{2})'}\right)\right)\be^{3}.
\end{array}
\end{equation}
Here $\alpha^2(k^{3})$, $\alpha^3(k^{3})$, $\beta(k^{3})$, $\gamma(k^{3})$ are arbitrary functions such that $(\alpha^{3})'\ne\beta(\alpha^{2})'$.

It is remain to observe the last case given by equations (\ref{Sep35}). Similar calculations shows, that this system does not have non-trivial solutions. Thus, investigation of the case $s=2$ is completed.

\section{The final list of solutions} Analysis of systems (\ref{Sep11})--(\ref{Sep35}) produces non-equivalent solutions given by the following formulae: (\ref{SBMHD:s0}), (\ref{SBMHD:slin}), (\ref{SBMHD:s4}), (\ref{SBMHD:s41}), (\ref{SBMHD:s11}), (\ref{solve:21}), (\ref{SBMHD:s3}), (\ref{Sol35}) and (\ref{sol351}). At that, solutions (\ref{SBMHD:s11}), (\ref{Sol35}) and (\ref{sol351}) represent special cases of formulae (\ref{SBMHD:s0}), (\ref{SBMHD:slin}). It is convenient to sort the obtained solutions in groups.

According to the representation of solution (\ref{SBMHD:14}) Maxwellian surfaces are translational, i.e. are sweeped by the parallel shift of the curve $\bx=\bsigma(k^1,k^3_0)$ along another curve $\bx=\btau(k^2,k^3_0)$ provided $k^3=k^3_0$ is fixed. Curves are parameterized by $k^1$ and $k^2$. Each of these curves can be a straight line (zero curvature), planar (zero torsion), or a general 3D curve. For convenience the curve with the lower ``dimension'' will be called the generatrix, whereas the curve with greater ``dimension'' will be called the directrix.

\begin{theorem}\label{T2}
Solutions of equation (\ref{SBMHD:15}) are exhausted by the following canonical representatives:

\begin{enumerate}
\item Solution with a straight generatrix, given in the general form by equations (\ref{SBMHD:s0}), (\ref{SBMHD:slin}).

\item Solutions with a planar generatrix, given by formulae (\ref{SBMHD:s4}), (\ref{SBMHD:s41}), (\ref{solve:21}).

\item Solutions with spatial generatrix, provided by formula (\ref{SBMHD:s3}).
\end{enumerate}
The complete set of solutions is obtained from the canonical ones by application of equivalence transformations (\ref{SBMHD:BT}) and (\ref{SBMHD:IT}).
\end{theorem}

The following theorem introduce additional classification on the set of solutions.

\begin{theorem}\label{T3}
All enumerated in Theorem \ref{T2} solutions are either specializations of solution (\ref{Sol21})--(\ref{Sol23}) with appropriately chosen vectors $\bsigma$ and $\balpha$, or solutions of the form (\ref{SBMHD:slin}).
\end{theorem}
\begin{proof}
For the proof it is required to demonstrate vectors $\bsigma$ and $\balpha$ in formulae (\ref{Sol21})--(\ref{Sol23}), that provide the coincidences of solutions.

1. Inclusion (\ref{SBMHD:s0}) $\subset$ (\ref{Sol21})--(\ref{Sol23}) is provided by the following choice of vectors
\[\bsigma=(k^1,0,0)^T,\quad\balpha=\bigl(1,\alpha^2(k^2,k^3), \alpha^3(k^2,k^3)\bigr)^T.\]
Hence, the second and the third components of equation (\ref{Sol23}) can be taken as definitions of functions $\alpha^2$ and $\alpha^3$. The first component of this equation gives the required restriction for functions $\tau^2$ and $\tau^3$.

2. Inclusion (\ref{SBMHD:s4}) $\subset$ (\ref{SBMHD:s41}). Let us take
\[n(k^3)=0,\quad f(\lambda)=k^1, \quad g(\lambda)=g(k^1).\]
in equation (\ref{SBMHD:s41}). Accurate to the rotation of coordinate axis this gives he solution (\ref{SBMHD:s4}).

3. Inclusion (\ref{SBMHD:s41}) $\subset$ (\ref{Sol21})--(\ref{Sol23}). Let us take
\[\bsigma=\bigl(f(\lambda),g(\lambda),0\bigr)^T,\quad\balpha=\bigl(-n(k^3),1, \alpha(k^2,k^3)\bigr)^T.\]
As before, an arbitrary function $\alpha(k^2,k^3)$ is given by the last component of equation (\ref{Sol23}). The remaining two components of the equation give the required restrictions to the vector $\btau$.

4. Inclusion (\ref{solve:21}) $\subset$ (\ref{Sol21})--(\ref{Sol23}). Let us choose
\[\bsigma=\bigl(u^1(k^1),u^2(k^1),k^1\bigr)^T,\quad\balpha=(0,0,1)^T.\]
Equation (\ref{Sol23}) has a non-trivial solution only if $\tau^3=0$. The required equation for the remaining components of vector $\btau$ follows from the third component of equation (\ref{Sol23}).

5. Inclusion (\ref{SBMHD:s3}) $\subset$ (\ref{Sol21})--(\ref{Sol23}). Here
\[\bsigma=\bigl(u^1(\lambda),u^2(\lambda),S(\lambda)\bigr)^T,\quad\balpha=\bigl(-\alpha(k^3),-\beta(k^3),1\bigr)^T.\]
In this case equation (\ref{Sol23}) for vector $\btau$ accurate to equivalence transformations has a solution, given by formulae (\ref{SBMHD:55}), (\ref{taurep}), (\ref{phieq}). Indeed, by scalar multiplication of equation (\ref{Sol23}) on $\btau_2$ and integration with respect to $k^2$ one obtains an integral
\[\btau(k^2,k^3)\cdot\balpha(k^3)=f(k^3).\]
It implies, that vector $\btau$ has the representation
\[\btau=\tau^1(k^2,k^3)\bepsilon^1(k^3)+\tau^2(k^2,k^3)\bepsilon^2(k^3)+\bkappa(k^3),\]
where
\begin{equation}\label{Eqstau}
\bepsilon^1\cdot\balpha=0,\quad\bepsilon^2\cdot\balpha=0,\quad\bkappa\cdot\balpha=f(k^3).
\end{equation}
By virtue of arbitrariness in the choice of functions $\btau^1$, $\btau^2$, equivalence transformations (\ref{SBMHD:47}) act on vectors $\bepsilon^1$, $\bepsilon^2$, and $\bkappa$. Due to these transformations, taking into account equations (\ref{Eqstau}) these vectors can be brought to a form (\ref{SBMHD:48}). In this representation equation (\ref{Sol23}) is equivalent to the system of equations (\ref{SBMHD:52a}). All the remaining considerations coincide with those used above for obtaining of the solution.
\end{proof}

\section{Conclusion} It is shown that all solutions of ideal MHD equations, that describe stationary flows with constant total pressure belong to classes given by formulae (\ref{SBMHD:slin}) and (\ref{Sol21})--(\ref{Sol23}). To prove this statement a curvilinear system of coordinates, in which streamlines and magnetic lines play a role of coordinate curves was used. It was shown, that Maxwellian surfaces, spanned by streamlines and magnetic lines, belong to a class of translational surfaces. Formulae (\ref{SBMHD:slin}) and (\ref{Sol21})--(\ref{Sol23}) provide restrictions of possible forms of directrix and generatrix of the surfaces.

Maxwellian surfaces determined by solution (\ref{SBMHD:slin}) are cylindrical. The directrix can be chosen with functional arbitrariness according to formulae (\ref{SBMHD:slin}), (\ref{SBMHD:slin1}). The direction of the generatrix can change while moving from one Maxwellian surface to another according to the modification of vector $\bfeta(k^3)$. Functional arbitrariness allows one to control the shape of streamlines and magnetic lines on Maxwellian surfaces.

For the family of solutions described by formulae (\ref{Sol21})--(\ref{Sol23}) all Maxwellian surfaces have common directrix, determined by the dependence $\bsigma(\lambda)$. Functional arbitrariness of the solution allows varying of the picture of the described fluid motion. In particular, solutions of this type describe stationary flows in arched jets, that can be used for description of solar prominences. For practical computations it is convenient to use special forms of solutions determined by formulae (\ref{SBMHD:s4}), (\ref{solve:21}).

The work is supported by RFBR (grant 11-01-00026), Presidential programs N.Sh-4368.2010.1 and MD-168.2011.1, and Federal Purpose Program 02.740.11.0617.

\appendix
\section{Separation of variables in equation (\ref{SBMHD:15})}\label{AppA} According to Lemma \ref{SBMHD:l2} and by virtue of symmetry of equation it is sufficient to observe cases $s=0,1,2$.

\underline{Case $s=0$.} Here matrix $C$ is absent. Let us express the 6-dimensional vector $\bd$ as a union of two 3-dimensional: $\bd=\bigl(\bmu(k^3),\bnu(k^3)\bigr)$. Hence,
\begin{equation}\label{SBMHD:18}
\bsigma_3\times\bsigma_1=\bmu(k^3),\quad\bsigma_1=\bnu(k^3).
\end{equation}
At that, in order to guarantee the dependence of vector $\bsigma$ on coordinate $k^1$, it is supposed that $\bnu\ne 0$. Note that accurate to equivalence transformation, this solution belong to class of solutions with cylindrical Maxwellian surfaces, described in section \ref{CylMax}. However, for the sake of completeness, this case will be observed separately.

The scalar multiplication of left- and right-hand sides of equations (\ref{SBMHD:18}) gives $\bmu\cdot\bnu=0$. This equation allows representation of vector $\bmu$
\[\bmu=\bar{\bmu}\times\bnu\]
with some vector $\bar{\bmu}(k^3)$. By plugging the obtained representation of vector $\bmu$ into the first equation of (\ref{SBMHD:18}) by virtue of the second equation one obtains
\[\bsigma_3\times\bsigma_1=\bar{\bmu}\times\bnu=\bar{\bmu}\times\bsigma_1\quad\Rightarrow\quad \bigl(\bsigma_3-\bar{\bmu}(k^3)\bigr)\times\bsigma_1=0.\]
Accurate to equivalence transformations one can take $\bar{\bmu}(k^3)=0$. Thus, $\bsigma_3\times\bsigma_1=0$, and system of equations (\ref{Sep11}) follows.

\underline{Case $s=1$.} Lemma \ref{SBMHD:l2} gives the following equations
\begin{equation}\label{SBMHD:22}
\bsigma_3\times\bsigma_1=U(k^1,k^3)\bmu(k^3)+\bnu(k^3), \quad\bsigma_1=U(k^1,k^3)\bepsilon(k^3)+\bkappa(k^3)
\end{equation}
with some auxiliary function $U(k^1,k^3)$ such that $\partial U/\partial k^1\ne0$, and 6-dimensional linearly independent vectors $(\bmu,\bepsilon)$ and $(\bnu,\bkappa)$. Cases of different dimension of the linear space $\{\bepsilon,\bkappa\}$ spanned by vectors $\bepsilon$ and $\bkappa$ for various values of parameter $k^3$ are observed.

{\sf 1) $\dim\{\bepsilon,\bkappa\}=2$.} By scalar multiplication of left- and right-hand sides of equations (\ref{SBMHD:22}) and by splitting on the different powers of function $u$ one obtains the following relations:
\begin{equation}\label{SBMHD:23}
\bmu\cdot\bepsilon=\bnu\cdot\bkappa=0,\quad\bmu\cdot\bkappa+\bnu\cdot\bepsilon=0.
\end{equation}
From the former two equations of system (\ref{SBMHD:23}) the representations follow:
\begin{equation}\label{SBMHD:27}
\bmu=\bbmu\times\bepsilon,\quad\bnu=\bbnu\times\bkappa
\end{equation}
with suitable vectors $\bbmu$ and $\bbnu$. By plugging this equalities into the latter equation of (\ref{SBMHD:23}) one obtains
\[(\bbmu\times\bepsilon)\cdot\bkappa+(\bbnu\times\bkappa)\cdot\bepsilon= 0\quad\Rightarrow\quad(\bbnu-\bbmu)\cdot(\bepsilon\times\bkappa)=0.\]
Hence,
\begin{equation}\label{SBMHD:28}
\bbnu=\bbmu+\alpha(k^3)\bepsilon+\beta(k^3)\bkappa
\end{equation}
with some functions $\alpha$, $\beta$. Since vectors $\bbmu$ and $\bbnu$ are defined accurate to additives, proportional to vectors $\bepsilon$ and $\bkappa$ respectively, in equations (\ref{SBMHD:28}) without limiting of generality one can assume $\alpha=\beta=0$. By using representations (\ref{SBMHD:27}) with $\bbnu=\bbmu$ in the first equation of (\ref{SBMHD:22}), taking into account the second equation of (\ref{SBMHD:22}), and collecting terms one obtains
\[(\bsigma_3-\bbmu)\times\bsigma_1=0.\]
By virtue of dependence of vector $\bbmu$ only on $k^3$ and due to the arbitrariness of vector $\bsigma$, without limiting the generality, one can assume $\bbmu=0$, i.e. $\bsigma_3\times\bsigma_1=0$. Substitution of obtained representations into equations (\ref{SBMHD:17}) and splitting on powers of function $U$ gives equations for vector $\btau$, specified in system (\ref{Sep21}).

{\sf 2) $\dim\{\bepsilon,\bkappa\}=1$.} Two cases are possible: either a) $\bkappa=\alpha(k^3)\bepsilon$, or b) $\bepsilon=0$ and $\bkappa\ne0$. In case a) by a suitable choice of function $U$ in representation (\ref{SBMHD:22}) one can make $\bkappa=0$, i.e.
\begin{equation}\label{SBMHD:22a}
\bsigma_1=U(k^1,k^3)\bepsilon.
\end{equation}
Dot product of obtained equation with the first of equations (\ref{SBMHD:22}), and splitting with respect to powers of function $U$ we obtain equalities, analogous to equations (\ref{SBMHD:23}):
\begin{equation}\label{SBMHD:23a}
\bmu\cdot\bepsilon=0,\quad\bnu\cdot\bepsilon=0.
\end{equation}
Equations (\ref{SBMHD:23a}) yield the representation
\begin{equation}\label{SBMHD:27a}
\bmu=\bbmu\times\bepsilon,\quad\bnu=\bbnu\times\bepsilon
\end{equation}
with suitable vectors $\bbmu$ and $\bbnu$. Substitution of representations (\ref{SBMHD:27a}) into the first of equations (\ref{SBMHD:22}) taking into account equation (\ref{SBMHD:22a}) leads to the relation $(\bsigma_3-\bbmu)\times\bsigma_1=\bbnu\times\bepsilon$. Arbitrariness in the choice of vector $\bsigma$ allows simplification $\bbmu=0$. Thus, for vector $\bsigma$ one obtains the following relations:
\begin{equation}\label{SBMHD:28a}
\bsigma_3\times\bsigma_1=\bbnu\times\bepsilon,\quad\bsigma_1=U(k^1,k^3)\bepsilon.
\end{equation}
Substitution of equation (\ref{SBMHD:28a}) into (\ref{SBMHD:17}) and splitting with respect to $U$ results equations (\ref{Sep22}) for vectors $\bsigma$ and  $\btau$.

Similar actions in case b) leads to equations (\ref{Sep23}).

\underline{Case $s=2$.} The following equations are to be observed:
\begin{eqnarray}\label{SBMHD:39}
&&\bsigma_3\times\bsigma_1=U^1(k^1,k^3)\bmu^1(k^3)+U^2(k^1,k^3)\bmu^2(k^3)+\bnu(k^3),\\[2mm]\label{SBMHD:40} &&\bsigma_1=U^1(k^1,k^3)\bepsilon^1(k^3)+U^2(k^1,k^3)\bepsilon^2(k^3)+\bkappa(k^3)
\end{eqnarray}
Here 6-dimensional vectors $(\bmu^1,\bepsilon^1)$, $(\bmu^2,\bepsilon^2)$, $(\bnu,\bkappa)$ are linearly independent for every value of parameter $k^3$; functions $U^1$, $U^2$ and $1$ are linearly independent as functions of $k^1$ (coefficients in linear combinations of functions can depend on $k^3$). The classifying parameter is the dimension of the linear space $\{\bepsilon^1,\bepsilon^2,\bkappa\}$ spanned by vectors $\bepsilon^1$, $\epsilon^2$, and $\bkappa$ for different values of parameter $k^3$.

{\sf 1) $\dim\{\bepsilon^1,\bepsilon^2,\bkappa\}=3$.} The dot product of left and right-hand sides of equations (\ref{SBMHD:39}), and (\ref{SBMHD:40}), splitting with respect to linearly independent functions $U^1$, $U^2$, and $1$ yield:
\begin{equation}\label{SBMHD:41}
\begin{array}{l}
\bmu^1\cdot\bepsilon^1=\bmu^2\cdot\bepsilon^2=\bnu\cdot\bkappa=0,\\[2mm]
\bmu^1\cdot\bepsilon^2+\bmu^2\cdot\bepsilon^1=0,\quad\bmu^1\cdot\bkappa+\bnu\cdot\bepsilon^1=0,\quad\bmu^2\cdot\bkappa+\bnu\cdot\bepsilon^2=0.
\end{array}
\end{equation}
From the former three equations of (\ref{SBMHD:41}) the representations follow:
\begin{equation}\label{SBMHD:42}
\bmu^1=\bbmu^1\times\bepsilon^1,\quad\bmu^2=\bbmu^2\times\bepsilon^2,\quad\bnu=\bbnu\times\bkappa
\end{equation}
with some vectors $\bbmu^i$, $\bbnu$. Substitution of these representations into the latter three equations of (\ref{SBMHD:41}) gives
\[\fl(\bbmu^1-\bbmu^2)\cdot(\bepsilon^1\times\bepsilon^2)=0,\quad(\bbmu^1-\bbnu)\cdot(\bepsilon^1\times\bkappa)=0,\quad
(\bbmu^2-\bbnu)\cdot(\bepsilon^2\times\bkappa)=0\]
Taking into account linear independence of vectors $\bepsilon^1$, $\bepsilon^2$, $\bkappa$, and the arbitrariness of vectors $\bbmu^i$, $\bbnu$, this implies
\[\bbmu^1=\bbmu^2=\bbnu.\]
Substitution of formulae (\ref{SBMHD:42}) with the preceding equations into (\ref{SBMHD:39}) after some simplification leads to
\[\bsigma_3\times\bsigma_1=\bbnu\times\bigl(U^1\bepsilon^1+U^2\bepsilon^2+\bkappa\bigr)=\bbnu\times\bsigma_1\quad
\Rightarrow\quad(\bsigma_3-\bbnu)\times\bsigma_1=0.\]
Since vector $\bsigma$ is defined accurate to the additive vector, depending only on $k^3$, without restriction of generality one can assume $\bbnu=0$.

Substitution of the obtained representation of vector $\bsigma$ into equation (\ref{SBMHD:17}) and splitting with respect to the remaining linearly independent functions $U^1$, $U^2$, and $1$ produces equations (\ref{Sep31}).

{\sf 2) $\dim\{\bepsilon^1,\bepsilon^2,\bkappa\}=2$.} The following subcases are possible: a) vector $\bkappa$ can be linearly expressed in terms of vectors $\bepsilon^1$ and $\bepsilon^2$, or b) vectors $\bepsilon^1$ and $\bepsilon^2$ are proportional or one of them is equal to zero.

In case a) by addition of terms, depending only on  $k^3$ to $U^1$ and $U^2$, one can make $\bkappa=0$. Dot product of left- and right-hand sides of equations (\ref{SBMHD:39}), (\ref{SBMHD:40}) and splitting with respect to functions $U^i$ leads to equations (\ref{SBMHD:41}) with $\bkappa=0$. These equations, in turn, give
\begin{equation}\label{SBMHD:42a}
\bmu^1=\bbmu\times\bepsilon^1,\quad\bmu^2=\bbmu\times\bepsilon^2,\quad\bnu=\delta(k^3)\bepsilon^1\times\bepsilon^2
\end{equation}
with some vector $\bbmu(k^3)$ and function $\delta(k^3)$. Substitution of obtained representations into equation (\ref{SBMHD:39}) and action of equivalence transformations lead to equations (\ref{Sep32}).

In case b) by a proper non-degenerate change of functions $U^1$, and $U^2$ it is always possible to make $\bepsilon^2=0$ in equation (\ref{SBMHD:40}). In this case relations (\ref{SBMHD:41}) produce representations for vectors $\bmu^i$, and $\bnu$ in the form
\begin{equation}\label{SBMHD:42b}
\bmu^1=\bbmu\times\bepsilon^1,\quad\bmu^2=\alpha\bepsilon^1\times\bkappa,\quad\bnu=\bbmu\times\bkappa
\end{equation}
with some vector $\bbmu$. The system of equations (\ref{Sep33}) is obtained in the same way as in previous case.

{\sf 3) $\dim\{\bepsilon^1,\bepsilon^2,\bkappa\}=1$.} The following subcases are possible: a) vectors $\bkappa$ and $\bepsilon^2$ linearly related with vector $\bepsilon^1$, or b) vectors are zero: $\bepsilon^1=\bepsilon^2=0$. By the analogous calculations one can show that case a) gives a system of equations (\ref{Sep34}), and case b) reduces to the system (\ref{Sep35}).

\end{document}